%% file: Main.tex
\newif\ifprocs
\procstrue
\procsfalse

\ifprocs

\else
\documentclass[11pt,USenglish]{article}
\usepackage[margin=1.0in]{geometry}
\usepackage[T1]{fontenc}
\usepackage[utf8]{inputenc}
\usepackage[english]{babel}
\fi

\usepackage{mwe}
\usepackage{caption}
\usepackage{xcolor}
\usepackage{algorithm}
\usepackage{amsmath}
\usepackage{amsthm}
\usepackage{amssymb}
\usepackage{complexity}
\usepackage{graphicx}
\usepackage{amsfonts}
\usepackage{xspace}
\usepackage[export]{adjustbox}

\ifprocs
\else
\usepackage{ifpdf}
\ifpdf
\usepackage[pdftex,colorlinks,allcolors=blue]{hyperref}
\else
\usepackage[hypertex,colorlinks,linkcolor=blue,citecolor=blue,filecolor=blue,urlcolor=blue]{hyperref}
\fi
\fi

\ifprocs
\else
\newtheorem{theorem}{Theorem}[section]
\newtheorem{lemma}[theorem]{Lemma}
\newtheorem{assumption}[theorem]{Assumption}

\fi

\theoremstyle{plain}

\makeatletter
\newtheorem*{rep@theorem}{\rep@title}
\newcommand{\newreptheorem}[2]{%
\newenvironment{rep#1}[1]{%
 \def\rep@title{#2 \ref{##1}}%
 \begin{rep@theorem}}%
 {\end{rep@theorem}}}
\makeatother

\newreptheorem{theorem}{Theorem}
\newreptheorem{lemma}{Lemma}

\newcommand{\ProblemName}[1]{\textsf{#1}}

\newcommand{\SCO}{\ProblemName{Set Cover}\xspace}

\newcommand{\PPC}{\ProblemName{p-Partial Cover}\xspace}

\newcommand{\SI}{\ProblemName{Subgraph Isomorphism}\xspace}
\newcommand{\KTR}{\ProblemName{kTree}\xspace}

\providecommand{\card}[1]{\lvert#1\rvert}

\title{Conditional Lower Bound for Subgraph Isomorphism with a Tree Pattern%
\thanks{This work was partially supported by 
a Minerva Foundation grant.}
\thanks{A merged work containing the results in this paper is available on \href{url}{http://arxiv.org/abs/1711.08041}.}
}

\ifprocs

\else
\author{Robert Krauthgamer\thanks{Email: \texttt{robert.krauthgamer@weizmann.ac.il}  }
\qquad 
Ohad Trabelsi\thanks{Email: \texttt{ohad.trabelsi@weizmann.ac.il}}
\\
Weizmann Institute of Science
}

\fi

\begin{document}
\maketitle

\input{Intro}

\input{Paper}

{
\ifprocs
\else
\small
\bibliographystyle{alphaurlinit}
\bibliography{robi}
}

\end{document}

%% file: Intro.tex
\begin{abstract}
The \KTR problem is a special case of \SI where the pattern graph is a tree,
that is, the input is an $n$-node graph $G$ and a $k$-node tree $T$, 
and the goal is to determine whether $G$ has a subgraph isomorphic to $T$.
We provide evidence that this problem cannot be computed significantly faster than $2^{k}\poly(n)$, 
which matches the fastest algorithm known for this problem by Koutis and Williams
[ICALP 2009 and TALG 2016]. Specifically, we show that if \KTR can be solved in time $(2-\varepsilon)^k\poly(n)$ for some constant $\varepsilon>0$, then \SCO with $n'$ elements and $m'$ sets can be solved in time $(2-\delta)^{n'}\poly(m')$ for a constant $\delta(\varepsilon)>0$, which would refute the Set Cover Conjecture by Cygan et al. [CCC 2012 and TALG 2016]. 

Our techniques yield a new algorithm for the \PPC problem, 
a parameterized version of \SCO that requires covering at least $p$ elements (rather than all elements).
Its running time is $(2+\varepsilon)^p (m')^{O(1/\varepsilon)}$
for any fixed $\varepsilon>0$, 
which improves the previous $2.597^p\poly(m')$-time algorithm by Zehavi [ESA 2015]. Our running time is nearly optimal, as a $(2-\varepsilon')^p\poly(m')$-time algorithm would refute the Set Cover Conjecture.
\end{abstract}

\section{Introduction}

The \SI problem was studied extensively in theoretical computer science. The most basic version of it asks whether a host graph $G$ contains a copy of a pattern graph $H$ as a subgraph. It is well known to be NP-hard since it generalizes hard problems such as \ProblemName{Maximum Clique} and \ProblemName{Hamiltonicity}~\cite{karp1972},
but unlike many natural NP-hard problems, it requires $N^{\Omega(N)}$ time where $N=\card{V(G)}+\card{V(H)}$ is the total number of vertices, unless the exponential time hypothesis (ETH) fails~\cite{CFGKMP16}. 
Hence, most past research addressed its special cases that are in $P$, including the case where the pattern graph is of constant size~\cite{marx14}, or when both graphs are trees~\cite{AbboudWY15}, biconnected outerplanar graphs~\cite{Lingas89}, two-connected series-parallel graphs~\cite{lovasz2009}, and more~\cite{Dessmark00, Matousek92}. 

We will focus on the version where the pattern is a tree $T$ on $k$ nodes, and the goal is to decide whether $G$ contains a copy of $T$ as a subgraph. For this special case, called \KTR, a couple of different techniques were used in order to design algorithms. The color-coding method, designed by Alon, Yuster, and Zwick~\cite{Alon95}, yields an algorithm with running time $O^*((2e)^{k})$,
where throughout, $O^*(\cdot)$ hides polynomial factors in the instance size.
Later, a new method that was developed utilizes \ProblemName{kMLD} 
(stands for $k$ Multilinear Monomial Detection -- the problem of detecting multilinear monomials of degree $k$ in polynomials presented as circuits) 
to create a \KTR algorithm with running time $O^*(2^k)$~\cite{Koutis16}. 
Our main result shows that this running time is actually optimal 
(up to exponential improvements) 
based on the Set Cover Conjecture (SeCoCo), introduced by~\cite{cygan16}. 
In the \SCO problem, the input is a ground set $[n]=\{1,...,n\}$ 
and a collection of $m$ sets, and the goal is to find the smallest sub-collection of sets whose union is the entire ground set. 
SeCoCo implies that it cannot be solved significantly faster than $O^*(2^{n})$. 
We can now state our main result,
whose proof appears in Section~\ref{ProofsKTR}. 

\begin{theorem}\label{Thm:LB}
If for some fixed $\varepsilon>0$, \KTR can be solved in time $O^*((2-\varepsilon)^k)$, then for some $\delta(\varepsilon)>0$, \SCO on $n$ elements and $m$ sets can be solved in time $O^*((2-\delta)^n)$.
\end{theorem}

In spite of extensive effort, the fastest algorithm for \SCO is still 
essentially the folklore dynamic programming algorithm that runs in time $O^*(2^n)$, with several improvements in special cases~\cite{Koivisto09, Bjorklund09, neder16, bjor17}. SeCoCo states that for every fixed $\varepsilon>0$ there is an integer $t(\varepsilon)>0$ such that \SCO with sets of size at most $t$ cannot be computed in time $O^*(2^{(1-\varepsilon)n})$. It clearly implies that for every fixed $\varepsilon>0$, \SCO cannot be solved in time $O^*(2^{(1-\varepsilon)n})$.

Several conditional lower bounds were based on this conjecture in the recent decade, including for \ProblemName{Set Partitioning}, \ProblemName{Connected Vertex Cover}, \ProblemName{Steiner Tree}, \ProblemName{Subset Sum}~\cite{cygan16} (though for the last problem, it was later proved assuming instead Strong ETH (SETH)~\cite{abboud2017seth}), \ProblemName{Maximum Graph Motif}~\cite{bjor16}, parity of the number of solutions to \SCO with at most $t$ sets~\cite{bjor15}, \ProblemName{Colorful Path} (interestingly, this problem is a sub-routine in the aforementioned color-coding method, however, this lower bound says nothing about the \KTR problem) and \ProblemName{Colorful Cycle}~\cite{kow16}, and the dynamic, general and connected versions of \ProblemName{Dominating Set}~\cite{kri17}.  

Note that our conditional lower bound is for the undirected version of \KTR. The directed version of \KTR is defined similar to the undirected version, except that $G$ is a directed graph, and $T$ is a directed tree, that is the undirected version of $T$ is a tree. This directed version of \KTR can only be harder - even when the directed tree is an arborescence, as one can reduce the undirected version to it with only a polynomial loss, as follows. Define the host graph $G'$ to be $G$ with edges in both directions, and direct the edges in $T$ away from an arbitrary vertex $v\in T$ to create the directed tree $T'$, which is thus an arborescence. Clearly, the directed instance is a yes-instance if and only if the undirected instance also is.

\medskip
Our techniques yield a new algorithm for the \PPC problem, 
whose input is similar to the \SCO problem but with an additional integer $p$, 
and the goal is to find the smallest sub-collection of sets whose union contains at least $p$ elements. 
The previously fastest algorithm for this problem
runs in time $O^*(2.597^p)$~\cite{Zehavi2015}. 
\begin{theorem}\label{Thm:Alg}
For every fixed $\varepsilon>0$, \PPC can be solved in time $(2+\varepsilon)^p m^{O(1/\varepsilon)}$.
\end{theorem}
Our proof is based on simple modifications to the reduction of Theorem~\ref{Thm:LB} and appears in Section~\ref{ProofsPPC}. 
Observe that an $O^*((2-\varepsilon')^p)$-time algorithm for \PPC would violate SeCoCo (because \SCO is a special case of \PPC with $p=n$), hence our algorithm's running time is almost optimal. 


\subparagraph*{Prior Work.} 
As mentioned earlier, for the general version of \SI, 
no algorithm can decide whether a host $N$-vertex graph $G$ contains a subgraph isomorphic to a pattern $N$-vertex graph $H$ in time $N^{o(N)}$, 
unless the ETH fails~\cite{CFGKMP16}. 
For the version where both the host and the pattern graphs are rooted trees of size $N$, a tight lower bound of $N^{2-o(1)}$ was proved~\cite{Abboud16} based on the Orthogonal Vectors Conjecture (and consequently on SETH). 
The \KTR problem cannot be solved in time $2^{o(k)}$ assuming the ETH,
because directed Hamiltonicity is a special case of this problem and 
there is a simple reduction (with polynomial blowup) from \ProblemName{3SAT}. 
If we care about the exact exponent in the running time, 
only a restricted lower bound is known.
Williams and Koutis~\cite{Koutis16} used communication complexity to show that a faster algorithm for their intermediate problem \ProblemName{kMLD} in some settings is not possible, and so among a specific class of algorithms, 
their $O^*(2^k)$ algorithm for \KTR is optimal. 

The exact running time of \PPC was first studied by Bl{\"a}ser~\cite{blaser03},
who provided a randomized $O^*(5.437^p)$-time algorithm.
This was followed by a deterministic $O^*(4^p p^{2p})$-time algorithm~\cite{Bonnet2013},
and both were improved to a deterministic $O*(2.619^p)$-time algorithm~\cite{Sha16}.
Finally, the aforementioned deterministic $O^*(2.597^k)$-time algorithm
was devised by Zehavi~\cite{Zehavi2015}.


%% file: Paper.tex
\section{Reduction to \KTR}\label{ProofsKTR}

In this section we prove Theorem~\ref{Thm:LB}. In order to make the proof simpler, we will have a couple of assumptions regarding the \SCO instance. First, for a constant $g>0$ to be determined later, 
we can assume that all the sets in the \SCO instance are of size at most $n/g^2$, 
and that the optimal solution has at least $g$, as these cases can already be solved significantly faster than $O^*(2^n)$, proving the theorem for them in a degenerate manner. 
We formalize it as follows.

\begin{assumption}\label{asm1}
\textit{All the sets in the \SCO instance are of size at most $n/g^2$.}
\end{assumption}
To justify this assumption, notice that if some optimal solution for the \SCO instance contains a set of size at least $n/g^2$, we can find such optimal solution by simply guessing one set of at least this size (using exhaustive search over at most $m$ choices) and then applying the known dynamic programming algorithm on the still uncovered elements (at most $n-n/g^2$ of them), and return the optimal solution in total time $O^*(2^{(1-1/g^2)n})$.

\begin{assumption}\label{asm2}
\textit{No solution has size less than $g$.}
\end{assumption}
The reason that this assumption can be made is that if some optimal solution for the \SCO instance contains at most $g-1$ sets, then it is easy to find it in polynomial time, as the number of possibilities is $O(m^g)$, and $g$ is a constant. We continue to the following lemma, which is the heart of the proof.

\begin{lemma}\label{LemmaKTR}
For every fixed $\varepsilon>0$, \SCO on a ground set $N=[n]$ and a collection $M$ of $m$ sets that satisfies assumptions~\ref{asm1} and~\ref{asm2}, 
can be reduced to $2^{O(\sqrt{n})}$ instances of \KTR with $k=(1+\varepsilon)n+O(1)$.
\end{lemma}

\begin{proof}[Proof of Theorem~\ref{Thm:LB}]
Assuming that for some $\varepsilon\in (0,1)$, \KTR can be solved in time $O^*((2-\varepsilon')^k)\leq O^*(2^{(1-\varepsilon'/2)k})$. We reduce the \SCO instance by applying Lemma~\ref{LemmaKTR} 
with $\varepsilon=\varepsilon'/4$, 
and then solve each of the $2^{c_1\sqrt{n}}$ instances of \KTR 
in the assumed time of $O^*(2^{(1-\varepsilon'/2)((1+\varepsilon)n+c_2)})$, 
where $c_1,c_2>0$ are the constants implicit in the terms 
$2^{O(\sqrt{n})}$ and $O(1)$ in the lemma, respectively.
The total running time is 
$ O^*(2^{(1-\varepsilon'/2)(1+\varepsilon)n+c_1\sqrt{n}})
  = O^*(2^{(1-\varepsilon'/4-\varepsilon'^2/8)n+c_1\sqrt{n}})
  \leq O^*(2^{(1-\varepsilon'/4)n})\leq O^*((2-\varepsilon'/4)^n)$, 
which concludes the proof for $\delta(\varepsilon')=\varepsilon'/4$.

\end{proof}

To outline the proof of Lemma~\ref{LemmaKTR}, 
we will need the following definition. 
For an integer $a>0$, let $p(a)$ be the set of all unordered partitions of $a$, where a \emph{partition} of $a$ is a way of writing $a$ as a sum of positive integers, 
and \emph{unordered} means that the order of the summands is insignificant. 
The asymptotic behaviour of $\card{p(a)}$ (as $a$ tends to infinity) 
is known~\cite{hardy1918} to be
$$
  e^{\pi \sqrt{{2a}/{3}}}/(4a\sqrt{3})=2^{O(\sqrt{a})}. 
$$
It is possible to enumerate all the partitions of $a$ with constant delay between two consecutive partitions, 
exclusive of the output~\cite[Chapter 9]{nijen78}.

Now the intuition for our reduction of \SCO to \KTR
is that we first guess a partition of $n$ (the number of elements)
that represents how an optimal solution covers the elements as follows ---
associate each element arbitrarily with one of the sets that contain it
(so in effect, we assume each element is covered only once) 
and count how many elements are covered by each set in the optimal solution.
This guessing is done by exhaustive search over $p(n)\le 2^{O(\sqrt{n})}$ 
partitions of $n$. 
Then, we represent the \SCO instance using a \SI instance,
whose pattern tree $T$ succinctly reflects the guessed partition of $n$.
The idea is that the tree is isomorphic to a subgraph of the \SCO graph 
if and only if the \SCO instance has a solution that agrees with our guess.
\begin{proof}[Proof of Lemma~\ref{LemmaKTR}]

Given a \SCO instance on $n$ elements $N=\{n_i: i \in [n]\}$ and $m$ sets $M=\{S_i\}_{i\in [m]}$ and an $\varepsilon>0$, we construct $2^{O(\sqrt{n})}$ instances of \KTR as follows. For a constant $g(\varepsilon)$ to be determined later, the host graph $G_g=(V_g,E_g)$ is the same for all the instances, and is built on the bipartite graph representation of the \SCO instance, with some additions. This is done in a way that a constructed tree will fit in $G_g$ if and only if the \SCO instance has a solution that corresponds to the structure of the tree, as follows (see Figure~\ref{figs:G}). The set of nodes is $V_g=N\cup M\cup M_g\cup R \cup \{r_{g}, r_1, r_2, r\}$, where 
$M_g=\{ X\subseteq M : \card{X}=g \}$ 
 and $R=\{v_j^i:i\in [4], j\in [n/(g/2)]\}$.
Intuitively, the role of $M_g$ is to keep the size of the trees small by representing multiple vertices in $M$ (multiple sets in \SCO) at once as the "powering" technique for \SCO done in~\cite{cygan16}\footnote{note that we can slightly simplify this step in the construction by using as a black box the equivalence from~\cite{cygan16} between solving \SCO in time $O^*(2^{(1-\varepsilon)n})$ and in time $O^*(2^{(1-\varepsilon')(n+t)})$ where $t$ is the solution size. However, we preferred to reduce directly from \SCO for compatibility with our parameters and generality reasons.}, 
and the role of $R$ and $\{r_{g}, r_1, r_2, r\}$ is to enforce that the trees 
we construct will fit only in certain ways.  

The set of edges is constructed as follows. 
Edges between $N$ and $M$ are the usual bipartite graph representation of \SCO 
(i.e., connect vertices $n_j\in N$ and $S_i\in M$ whenever $n_j\in S_i$).
We also connect vertex $X\in M_g$ to vertex $n_j\in N$ 
if at least one of the sets in $X$ contains $n_j$. 
Additionally, we add edges between $r_{g}$ and every vertex in $M_g$, and $v^4_j\in R$ for $j\in [n/(g/2)]$, between $r_i$ and $v^i_j$ for every $i\in \{1,2\}$ and $j\in [n/(g/2)]$, and finally between $r$ and every vertex $v\in \{r_{g}, r_1, r_2\}$, $S_i\in M$, and $v^3_j\in R$ for $j\in [n/(g/2)]$.

\begin{figure}[!ht]
	\centering
		\includegraphics[width=1.0\textwidth,left]{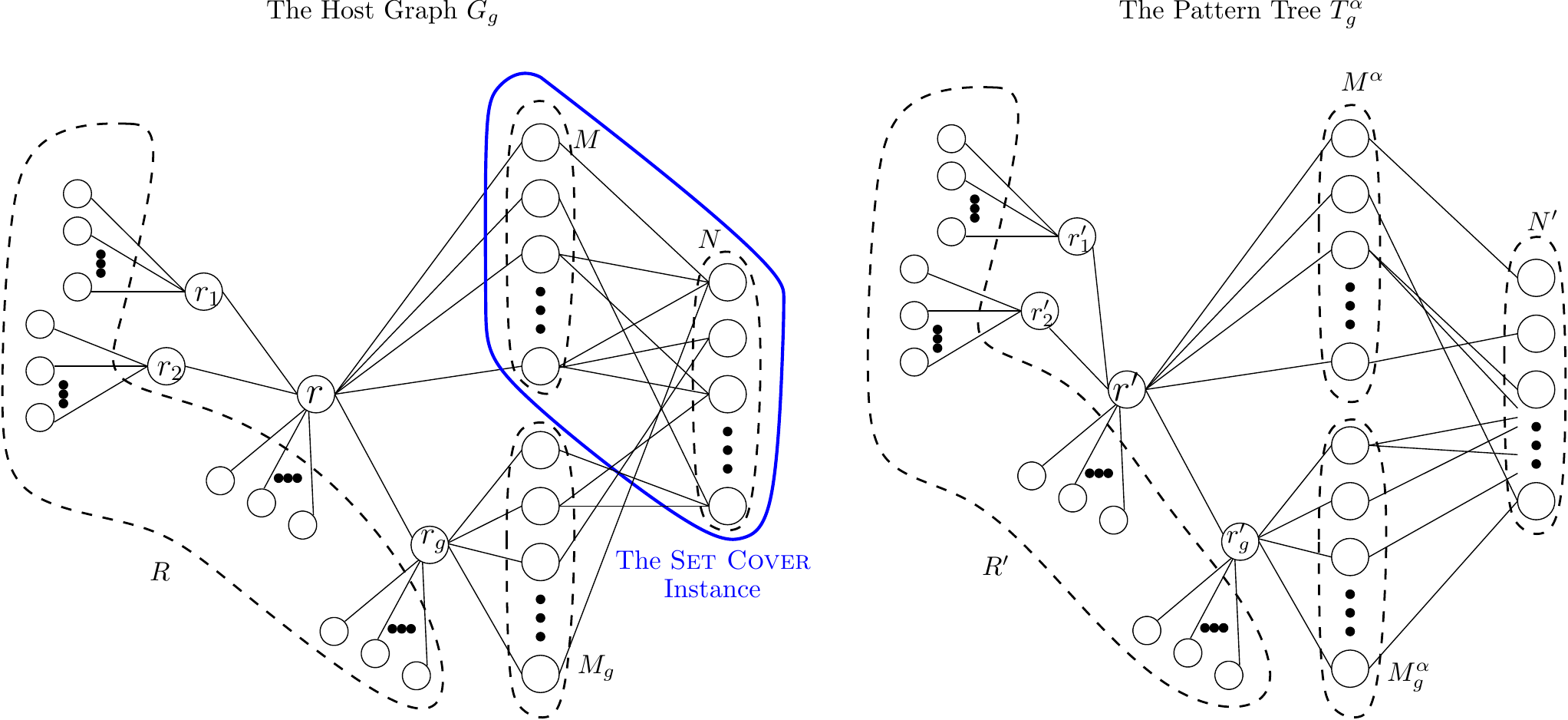}
   \caption[-]{An illustration of part of the reduction. The \SCO instance is depicted in blue, and sets of vertices are indicated by dashed curves.
   }
   \label{figs:G}
\end{figure}

Next, we construct $2^{O(\sqrt{n})}$ trees such that identifying those that are isomorphic to a subgraph of $G_g$ will determine the optimum of the \SCO instance. 

For every partition $\alpha=(p_1,p_2,...,p_l)\in p(n)$ (with possible repetitions) where $p(n)$ is as defined above, we construct a tree $T_g^{\alpha}=(V_g^{\alpha}, E_g^{\alpha})$. This tree has the same set of edges and vertices as $G_g$, except for the vertices in $M\cup M_{g}$ and the edges incident to them, which we replace by a set of new vertices $M^{\alpha}\cup M_g^{\alpha}$, and connect these new vertices to the rest in a way that the resulting graph is a tree. 
In more detail, $V_g^{\alpha}=N'\cup M^{\alpha}\cup M_g^{\alpha}\cup R'\cup \{r'_g, r'_1, r'_2, r'\}$ where $N', R', r'_g, r_1', r_2', r'$ are tagged copies of the originals, and $M^{\alpha}, M_g^{\alpha}$ are initialized to be $\emptyset$.

We define $\alpha_{g}$ to be a partition of $n$ which is also a shrinked representation of $\alpha$ by partitioning $\alpha$ into sums of $g$ numbers for a total of $\lfloor l/g \rfloor$ such sums, and a remaining of less than $g$ numbers. Formally,
$$\alpha_{g}=(\sum^g_{i=1} p_i, \sum^{2g}_{i=g+1} p_i,..., \sum^{g \lfloor l/g \rfloor}_{i=(g-1)\cdot \lfloor l/g \rfloor+1} p_i, p_{g \lfloor l/g \rfloor+1},...,p_l)
$$
Note that all the numbers in $\alpha_{g}$ are a sum of $g$ numbers in $\alpha$, except (maybe) for the last $g':=l-g\lfloor l/g \rfloor<g$ numbers in $\alpha_{g}$, a (multi)set which we denote $s(\alpha_{g})$. For every $i\in \alpha_{g}$ (with possible repetitions) we add a star on $i+1$ vertices to the constructed tree $T^{\alpha}_g$. If $i\in \alpha_{g}\setminus s(\alpha_g)$, we add the center vertex to $M^{\alpha}_g$, connect it to $r'_g$, and add the rest $i$ vertices to $N'$. Else, if $i\in s(\alpha_{g})$ we add the center vertex to $M^{\alpha}$, connect it to $r'$, and again add the rest $i$ vertices to $N'$. We return the minimum cardinality of $\alpha$ for which $(G_g, T_g^{\alpha})$ is a yes-instance. To see that this construction is small enough, note that the size of $G_g$ is at most $4+4\cdot n/(g/2)+m^g+m+n$ which is polynomial in $m$, and the size of the tree $T_g^{\alpha}$ is at most
$$
4+4\cdot n/(g/2)+n/g+g+n = n\cdot (1+9/g)+O(1) = n\cdot (1+\varepsilon)+O(1)
$$
where the last equality holds for $g=9/\varepsilon$, and so the size constraint follows.

We now prove that at least one of the trees $T_g^{\alpha}$ returns yes and satisfies $\card{\alpha}\leq d$, if and only if the \SCO instance has a solution of size at most $d$. For the first direction, assume that the \SCO instance has a solution $I$ with $\card{I}\leq d$. Consider a partition $\alpha_I\in p(n)$ of $n$ that corresponds to $I$ in the following way. Associate every element with exactly one of the sets in $I$ that contains it, and then consider the list of sizes of the sets in $I$ according to this association (eliminating zeroes). Clearly, $(G_g, T_g^{{\alpha_I}})$ is a yes-instance and so we will return a number that is at most $\card{I}$.

For the second direction, assume that every solution to the \SCO instance is of size at least $d+1$. We need to prove that for every tree $T^{\alpha}_g$ with $\card{\alpha}\leq d$, $(G_g,T_g^{\alpha})$ is a no-instance. Assume for the contrary that there exists such $\alpha$ for which $(G_g,T_g^{\alpha})$ is a yes-instance with the isomorphism function $f$ from $T_g^{\alpha}$ to $G_g$. We will show that the only way $f$ is feasible is if $f(r')=r$, $f(M^{\alpha})\subseteq M$, $f(M^{\alpha}_g)\subseteq M_g$, and also $f(N')= N$, which together allows us to extract a corresponding solution for the \SCO instance, leading to a contradiction. 
We start with the vertex $r'\in T_g^{\alpha}$. Since its degree is at least $n/(g/2)+3$ and by Assumption~\ref{asm1} and the construction of $G_g$, it holds that $f(r')\notin \{ r_1, r_2 \}\cup R\cup M\cup M_g $. Moreover, if it was the case that $f(r')\in \{ r_g \}\cup N$ then $\{ f(r'_1), f(r'_2)\}\cap (M\cup M_g)\neq \emptyset$, however, the degree of $r'_1$ and $r'_2$ in $T_g^{\alpha}$ is $n/(g/2)$, and the degree of the vertices in $M\cup M_g$ in $G_g$ is at most $g\cdot n/g^2=n/g$, so it must be that $f(r)=r$. 
Our next claim is that $f(r'_g)=r_g$. Observe that Assumption~\ref{asm2} implies that $M^{\alpha}_g\neq \emptyset$, and so $r'_g$ in the tree has vertices in distance $2$ from it and away from $r'$, a structural constraint that cannot be satisfied by any vertex in $\{ r_1,r_2 \}\cup R$. Furthermore, the degree of $r'_g$ is at least $n/(g/2)$ and so again by Assumption~\ref{asm1} it is also impossible that $f(r'_g)\in M^{\alpha}$, and hence it must be that $f(r'_g)=r_g$. Finally, by the same Assumption and the degrees of $r_1$ and $r_2$, $f(r'_1)$ and $f(r'_2)$ must be in $\{ r_1,r_2 \}$. Altogether, it must be that $f(M^{\alpha}_g)\subseteq M_g$, $f(M^{\alpha})\subseteq M$ and that $f(N')= N$, and therefore we can extract a feasible solution to the \SCO instance that has at most $d$ sets in it, which is a contradiction, concluding the proof of Lemma~\ref{LemmaKTR}. 




\end{proof}

\section{Algorithm for \PPC}\label{ProofsPPC}
In this section we prove Theorem~\ref{Thm:Alg}. To design an algorithm for \PPC, we reduce it to \KTR similarly to Lemma~\ref{LemmaKTR} with some adjustments, and then apply an algorithm for \KTR by~\cite{Koutis16}. First, as it will be enough to use the directed case of \KTR in order to get the desired bound, we simplify Lemma~\ref{LemmaKTR} by dealing with (directed) rooted trees, as follows. $G_g$ now contains only the nodes $N\cup M \cup M_g \cup \{r,r_g\}$ and the edges therein, directed away from the root $r$. Note that the two assumptions that precede Lemma~\ref{LemmaKTR} are not needed here, and that for any chosen $g$ the size of $V(G_g)$ is $O(m^g+m+n)$. We proceed to the description of the adjustments.

Instead of enumerating over all the partitions of $n$, we do it only for $p$ and hence the number of partitions is $2^{O(\sqrt{p})}$ with each partition $\alpha$ inducing a tree $T^{\alpha}_g$ in a similar way to Lemma~\ref{LemmaKTR}, of size at most $2p/g+p$. From here onwards, the proof of correctness is similar to Lemma~\ref{LemmaKTR}, and thus we omit it. By solving each \KTR instance in time $O(2^{2p/g+p}\card{V(G_g)}^{c_2})$ using the algorithm of~\cite{Koutis16}, where $c_2$ is the constant derived from there such that \KTR can be solved in time $O(2^k\card{V(G)}^{c_2})$, and setting $g=4/\varepsilon'$ for $\varepsilon'=\log_{2}(2+\varepsilon)-1$ where $\varepsilon\in(0,1)$ is the chosen parameter, we get a total running time of 
\begin{align*}
O(2^{2p/g+p+c_1\sqrt{p}}\cdot m^{c_2 g}) 
&= O(2^{\varepsilon'/2 \cdot p + p + c_1\sqrt{p}} \cdot m^{c_2 4/\varepsilon'})\\
&\leq O(2^{(1+\varepsilon')p} \cdot m^{c_2 4/\varepsilon'})\\
&\leq O(\cdot (2+\varepsilon)^p \cdot m^{c_2 8/\varepsilon}),
\end{align*}
since $\log(2+\varepsilon)-1\geq \varepsilon/2$ for $\varepsilon\in (0,1)$, and where $c_1$ is the constant implicit in the term $2^{O(\sqrt{p})}$, as required.


%% file: Main.bbl
\newcommand{\etalchar}[1]{$^{#1}$}
\begin{thebibliography}{ABH{\etalchar{+}}16}

\bibitem[ABH{\etalchar{+}}16]{Abboud16}
A.~Abboud, A.~Backurs, T.~D. Hansen, V.~V. Williams, and O.~Zamir.
\newblock Subtree isomorphism revisited.
\newblock In {\em Proceedings of the Twenty-seventh Annual ACM-SIAM Symposium
  on Discrete Algorithms}, SODA '16, pages 1256--1271. SIAM, 2016.
\newblock \href {http://dx.doi.org/10.1137/1.9781611974331.ch88}
  {\path{doi:10.1137/1.9781611974331.ch88}}.

\bibitem[ABHS17]{abboud2017seth}
A.~Abboud, K.~Bringmann, D.~Hermelin, and D.~Shabtay.
\newblock {SETH}-based lower bounds for subset sum and bicriteria path.
\newblock {\em CoRR}, 2017.
\newblock Available from: \url{http://arxiv.org/abs/1704.04546}.

\bibitem[AVY15]{AbboudWY15}
A.~Abboud, V.~{Vassilevska-Williams}, and H.~Yu.
\newblock Matching triangles and basing hardness on an extremely popular
  conjecture.
\newblock In {\em Proceedings of the Forty-seventh Annual ACM Symposium on
  Theory of Computing}, STOC '15, pages 41--50. ACM, 2015.
\newblock \href {http://dx.doi.org/10.1145/2746539.2746594}
  {\path{doi:10.1145/2746539.2746594}}.

\bibitem[AYZ95]{Alon95}
N.~Alon, R.~Yuster, and U.~Zwick.
\newblock Color-coding.
\newblock {\em J. ACM}, 42(4):844--856, July 1995.
\newblock \href {http://dx.doi.org/10.1145/210332.210337}
  {\path{doi:10.1145/210332.210337}}.

\bibitem[BHH15]{bjor15}
A.~Bj\"{o}rklund, D.~Holger, and T.~Husfeldt.
\newblock The parity of set systems under random restrictions with applications
  to exponential time problems.
\newblock In {\em 42nd International Colloquium on Automata, Languages and
  Programming (ICALP 2015)}, volume 9134, pages 231--242. Springer, 2015.
\newblock \href {http://dx.doi.org/10.1007/978-3-662-47672-7_19}
  {\path{doi:10.1007/978-3-662-47672-7_19}}.

\bibitem[BHK09]{Bjorklund09}
A.~Bj\"{o}rklund, T.~Husfeldt, and M.~Koivisto.
\newblock Set partitioning via inclusion-exclusion.
\newblock {\em SIAM J. Comput.}, 39(2):546--563, July 2009.
\newblock \href {http://dx.doi.org/10.1137/070683933}
  {\path{doi:10.1137/070683933}}.

\bibitem[BHPK17]{bjor17}
A.~Bj\"{o}rklund, T.~Husfeldt, K.~Ptteri, and M.~Koivisto.
\newblock Narrow sieves for parameterized paths and packings.
\newblock {\em Journal of Computer and System Sciences}, 87:119 -- 139, 2017.
\newblock \href {http://dx.doi.org/10.1016/j.jcss.2017.03.003}
  {\path{doi:10.1016/j.jcss.2017.03.003}}.

\bibitem[BKK16]{bjor16}
A.~Bj\"{o}rklund, P.~Kaski, and {\L}.~Kowalik.
\newblock Constrained multilinear detection and generalized graph motifs.
\newblock {\em Algorithmica}, 74(2):947--967, February 2016.
\newblock \href {http://dx.doi.org/10.1007/s00453-015-9981-1}
  {\path{doi:10.1007/s00453-015-9981-1}}.

\bibitem[Bl{\"a}03]{blaser03}
M.~Bl{\"a}ser.
\newblock Computing small partial coverings.
\newblock {\em Information Processing Letters}, 85(6):327--331, 2003.
\newblock \href {http://dx.doi.org/10.1016/S0020-0190(02)00434-9}
  {\path{doi:10.1016/S0020-0190(02)00434-9}}.

\bibitem[BPS13]{Bonnet2013}
E.~Bonnet, V.~T. Paschos, and F.~Sikora.
\newblock Multiparameterizations for max $k$-set cover and related
  satisfiability problems.
\newblock {\em CoRR}, 2013.
\newblock Available from: \url{http://arxiv.org/abs/1309.4718}.

\bibitem[CDL{\etalchar{+}}16]{cygan16}
M.~Cygan, H.~Dell, D.~Lokshtanov, D.~Marx, J.~Nederlof, Y.~Okamoto, R.~Paturi,
  S.~Saurabh, and M.~Wahlstr\"{o}m.
\newblock On problems as hard as {CNF-SAT}.
\newblock {\em ACM Transactions on Algorithms}, 12(3):41:1--41:24, May 2016.
\newblock \href {http://dx.doi.org/10.1145/2925416}
  {\path{doi:10.1145/2925416}}.

\bibitem[CFG{\etalchar{+}}16]{CFGKMP16}
M.~Cygan, F.~V. Fomin, A.~Golovnev, A.~S. Kulikov, I.~Mihajlin, J.~Pachocki,
  and A.~Soca{\l}a.
\newblock Tight bounds for graph homomorphism and subgraph isomorphism.
\newblock In {\em Proceedings of the Twenty-seventh Annual ACM-SIAM Symposium
  on Discrete Algorithms}, SODA '16, pages 1643--1649. SIAM, 2016.
\newblock \href {http://dx.doi.org/10.1137/1.9781611974331.ch112}
  {\path{doi:10.1137/1.9781611974331.ch112}}.

\bibitem[DLP00]{Dessmark00}
A.~Dessmark, A.~Lingas, and A.~Proskurowski.
\newblock Faster algorithms for subgraph isomorphism of $k$-connected partial
  $k$-trees.
\newblock {\em Algorithmica}, 27(3):337--347, January 2000.
\newblock \href {http://dx.doi.org/10.1007/s004530010023}
  {\path{doi:10.1007/s004530010023}}.

\bibitem[HR18]{hardy1918}
G.~H. Hardy and S.~Ramanujan.
\newblock Asymptotic formula{\ae} in combinatory analysis.
\newblock {\em Proceedings of the London Mathematical Society},
  s2-17(1):75--115, 1918.
\newblock \href {http://dx.doi.org/10.1112/plms/s2-17.1.75}
  {\path{doi:10.1112/plms/s2-17.1.75}}.

\bibitem[Kar72]{karp1972}
R.~M. Karp.
\newblock {\em Reducibility among Combinatorial Problems}, pages 85--103.
\newblock The IBM Research Symposia Series. Springer US, 1972.
\newblock \href {http://dx.doi.org/10.1007/978-1-4684-2001-2_9}
  {\path{doi:10.1007/978-1-4684-2001-2_9}}.

\bibitem[KL16]{kow16}
{\L}.~Kowalik and J.~Lauri.
\newblock On finding rainbow and colorful paths.
\newblock {\em Theoretical Computer Science}, 628(C):110--114, May 2016.
\newblock \href {http://dx.doi.org/10.1016/j.tcs.2016.03.017}
  {\path{doi:10.1016/j.tcs.2016.03.017}}.

\bibitem[Koi09]{Koivisto09}
M.~Koivisto.
\newblock Parameterized and exact computation.
\newblock chapter Partitioning into Sets of Bounded Cardinality, pages
  258--263. Springer-Verlag, 2009.
\newblock \href {http://dx.doi.org/10.1007/978-3-642-11269-0_21}
  {\path{doi:10.1007/978-3-642-11269-0_21}}.

\bibitem[KST17]{kri17}
R.~Krithika, A.~Sahu, and P.~Tale.
\newblock Dynamic parameterized problems.
\newblock In {\em 11th International Symposium on Parameterized and Exact
  Computation (IPEC 2016)}, volume~63 of {\em Leibniz International Proceedings
  in Informatics (LIPIcs)}, pages 19:1--19:14. Schloss
  Dagstuhl--Leibniz-Zentrum fuer Informatik, 2017.
\newblock \href {http://dx.doi.org/10.4230/LIPIcs.IPEC.2016.19}
  {\path{doi:10.4230/LIPIcs.IPEC.2016.19}}.

\bibitem[KW16]{Koutis16}
I.~Koutis and R.~Williams.
\newblock Limits and applications of group algebras for parameterized problems.
\newblock {\em ACM Transactions on Algorithms}, 12(3):31:1--31:18, May 2016.
\newblock \href {http://dx.doi.org/10.1145/2885499}
  {\path{doi:10.1145/2885499}}.

\bibitem[Lin89]{Lingas89}
A.~Lingas.
\newblock Subgraph isomorphism for biconnected outerplanar graphs in cubic
  time.
\newblock {\em Theoretical Computer Science}, 63(3):295--302, 1989.
\newblock \href {http://dx.doi.org/10.1016/0304-3975(89)90011-X}
  {\path{doi:10.1016/0304-3975(89)90011-X}}.

\bibitem[LP09]{lovasz2009}
L.~Lov{\'a}sz and M.~D. Plummer.
\newblock {\em Matching theory}, volume 367.
\newblock American Mathematical Society, 2009.

\bibitem[MP14]{marx14}
D.~Marx and M.~Pilipczuk.
\newblock {Everything you always wanted to know about the parameterized
  complexity of Subgraph Isomorphism (but were afraid to ask)}.
\newblock In {\em 31st International Symposium on Theoretical Aspects of
  Computer Science (STACS 2014)}, volume~25 of {\em Leibniz International
  Proceedings in Informatics (LIPIcs)}, pages 542--553. Schloss
  Dagstuhl--Leibniz-Zentrum fuer Informatik, 2014.
\newblock \href {http://dx.doi.org/10.4230/LIPIcs.STACS.2014.542}
  {\path{doi:10.4230/LIPIcs.STACS.2014.542}}.

\bibitem[MT92]{Matousek92}
J.~Matou{\v{s}}ek and R.~Thomas.
\newblock On the complexity of finding iso- and other morphisms for partial
  $k$-trees.
\newblock {\em Discrete Mathematics}, 108(1):343 -- 364, 1992.
\newblock \href {http://dx.doi.org/10.1016/0012-365X(92)90687-B}
  {\path{doi:10.1016/0012-365X(92)90687-B}}.

\bibitem[Ned16]{neder16}
J.~Nederlof.
\newblock Finding large set covers faster via the representation method.
\newblock In {\em 24th Annual European Symposium on Algorithms (ESA 2016)},
  volume~57 of {\em Leibniz International Proceedings in Informatics (LIPIcs)},
  pages 69:1--69:15. Schloss Dagstuhl--Leibniz-Zentrum fuer Informatik, 2016.
\newblock \href {http://dx.doi.org/10.4230/LIPIcs.ESA.2016.69}
  {\path{doi:10.4230/LIPIcs.ESA.2016.69}}.

\bibitem[NW78]{nijen78}
A.~Nijenhuis and H.~S. Will.
\newblock {\em Combinatorial Algorithms: For Computers and Hard Calculators}.
\newblock Academic Press, Inc., 2nd edition, 1978.

\bibitem[SZ16]{Sha16}
H.~Shachnai and M.~Zehavi.
\newblock Representative families.
\newblock {\em J. Comput. Syst. Sci.}, 82(3):488--502, May 2016.
\newblock \href {http://dx.doi.org/10.1016/j.jcss.2015.11.008}
  {\path{doi:10.1016/j.jcss.2015.11.008}}.

\bibitem[Zeh15]{Zehavi2015}
M.~Zehavi.
\newblock {\em Mixing Color Coding-Related Techniques}, pages 1037--1049.
\newblock Springer Berlin Heidelberg, 2015.
\newblock \href {http://dx.doi.org/10.1007/978-3-662-48350-3_86}
  {\path{doi:10.1007/978-3-662-48350-3_86}}.

\end{thebibliography}
